%% file: soda16.tex
\documentclass[11pt]{article}

\usepackage[margin=1in]{geometry}
\usepackage{mathtools}
\usepackage{mathptmx}
\usepackage{helvet}
\usepackage{courier}
\usepackage{type1cm}

\usepackage{makeidx}         
\usepackage{graphicx}        
\usepackage{multicol}        
\usepackage[bottom]{footmisc}

\usepackage{graphicx}
\usepackage{subcaption}
\usepackage[space]{grffile}
\usepackage{latexsym}
\usepackage{amsfonts,amsmath,amssymb}
\usepackage{url}

\usepackage[utf8]{inputenc}

\usepackage{fancyref}
\usepackage{hyperref}
\hypersetup{colorlinks=false,pdfborder={0 0 0},}
\usepackage{textcomp}
\usepackage{longtable}
\usepackage{multirow,booktabs}

\usepackage{algorithm}
\usepackage[noend]{algpseudocode}
\usepackage{algorithmicx}
\algnewcommand\True{\textbf{true}\space}
\algnewcommand\False{\textbf{false}\space}
\algnewcommand\Null{\textbf{null}\space}
\algnewcommand\AND{\textbf{and}\space}
\algnewcommand\OR{\textbf{or}\space}
\algnewcommand\NOT{\textbf{not}\space}
\algnewcommand\Continue{\textbf{continue}\space}
\algnewcommand\PlusEquals{$\mathrel{+}=$\space}
\algnewcommand\MinusEquals{$\mathrel{-}=$\space}

\usepackage{framed}
\usepackage[square,numbers]{natbib}
\usepackage{amsthm}

\usepackage{verbatim} 

\makeindex             

\newtheorem{lemma}{Lemma}

\DeclarePairedDelimiter\ceil{\lceil}{\rceil}
\DeclarePairedDelimiter\floor{\lfloor}{\rfloor}

\DeclareUnicodeCharacter{00A0}{~}

\newcommand{\Ex}{{\mathbf{Ex}}}
\renewcommand{\Pr}{{\mathbf{Pr}}}

\begin{document}

\date{}
\author{Donggu Kang\\MIT\\donggu@mit.edu \and James Payor\\MIT\\payor@mit.edu}
\title{Flow Rounding}
\maketitle

\input{Abstract.tex}

\setcounter{tocdepth}{1}

\input{introduction.tex}
\input{Background.tex}

\input{Rounding.tex}

\input{Rounding_mlgn.tex}
\input{Rounding_n2.tex}
\input{Rounding_mlgn2m.tex}

\section{Acknowledgement}
We thank Richard Peng for introducing the authors to the problem of flow rounding, communication of the $O(m \log n)$ algorithm, and fruitful discussions.

\nocite{*}
\bibliographystyle{apalike}
\bibliography{soda16}

\end{document}

%% file: Abstract.tex
\begin{abstract}

We consider \emph{flow rounding}: finding an integral flow from a fractional flow.
\emph{Costed flow rounding} asks that we find an integral flow with no worse cost.
\emph{Randomized flow rounding} requires we randomly find an integral flow such that the expected flow along each edge matches the fractional flow. Both problems are reduced to \emph{cycle canceling}, for which we develop an $O(m \log \frac {n^2} {m})$ algorithm.

\end{abstract}

%% file: introduction.tex
\section{Introduction}
\label{sec:intro}

Many modern network flow algorithms give solutions with fractional flow values, but often we're interested in an integral assignment of flow.
For any costed fractional flow with integral capacities, we can always change it to an integral flow with the same (or better) flow value and no worse cost. Further, we can always randomly find an integral solution so that the expected flow on each edge matches the fractional solution.

Given a fractional solution as a starting point, we refer to finding an integral flow with no worse cost as \emph{costed flow rounding}, and \emph{preserving} the expected amount of flow along each edge as \emph{randomized flow rounding}. Several authors have given $O(m \log m)$ algorithms for costed flow rounding - notably a scaling algorithm that has $O(\log m)$ parallel runtime \cite{Cohen1995}, approaches for rounding max-flow solutions using just random walks \cite{Lee2013} \cite{Madry2013}.
However, beyond a special case in \cite{Raghavan1987} where we round a unit acyclic flow, only the $O(m^2)$ path decomposition of \cite{Raghavan1985} is known for randomized flow rounding.

Our approach relies on an observation about edges with fractional flow: in a \emph{circulation}, every fractional edge must be in a cycle of fractional edges. Formalized in section \ref{sec:rounding}, this means we can repeatedly find and cancel \emph{fractional cycles} to yield an integral circulation.

When canceling a given cycle, we can push flow in either direction. In costed flow rounding, at least one direction won't increase overall cost. To approach randomized flow rounding, section \ref{sec:rounding} will show that simple random choices of direction preserve the expected flow along each edge.

For the \emph{cycle cancelling} problem, our main results are a practical $O(n^2)$ algorithm, and an $O(m \log \frac {n^2} {m})$ algorithm that smoothly combines the $O(n^2)$ approach with an $O(m \log n)$ dynamic tree solution given in \cite{Sleator1981}. This answers the conjecture of Goldberg and Tarjan \cite{Goldberg1989}, improving the speed of their minimum-mean cycle-cancelling algorithm and other rounding approaches involving cycle cancelling such as \cite{Fleischer00}.

We begin with brief background on the dynamic tree data structure of \cite{Sleator1981} and flows, followed by the reductions of flow rounding to cycle canceling, and then present each cycle canceling algorithm in turn.

%% file: Background.tex
\newpage
\section{Background}

\subsection{Dynamic trees}
\label{sec:linkcut}

Both the $O(m \log n)$ and $O(m \log \frac {n^2} {m})$ algorithms presented in this paper utilize the dynamic trees data structure of \cite{Sleator1981}. The original paper can provide greater detail, but we recall the relevant points.

The data structure allows the maintenance of a dynamic forest of rooted trees over a set of nodes, and tree-path operations over the forest. Each relevant operation takes time logarithmic size of the trees involved. Generally, this means we assume $O(\log n)$ time per operation, but when we restrict the size of each tree to $k$ in section \ref{sec:mlogn2m} this will guarantee operations take $O(\log k)$ time.

Specific supported operations we will need are:
\begin{itemize}
	\item $Link(u,v)$: Create an edge between $u$ and $v$, making $v$ a child of $u$.
	\item $Cut(u,v)$: Remove the edge between $u$ and $v$.
	\item $FindRoot(v)$: Find the root of a tree which $v$ belongs to.

	\item $PathAdd(u,v,c)$: Add a number $c$ to the weight of every edge along the $u$-$v$ path.
	\item $PathMin(u,v)$: Report the edge with minimum weight over the $u$-$v$ path. Break ties by reporting the edge closest to $u$.
	\item $PathSum(u,v)$: Report the sum of edge weights over the $u$-$v$ path.
\end{itemize}

\subsection{Flows and circulations}

For the remainder of the paper, we will focus on \emph{circulations}, flows in which no node has excesses and deficits. This is justified by reducing other cases to circulations.

In the case of costed flow rounding, we will be given a fractional circulation over a costed graph, and attempt to find an integral circulation with no worse cost. Note that this allows us to round max-flow and min-cost max-flow solutions as well: we connect the sink to the source with an edge that has cost $-\infty$ to find a circulation, and this edge guarantees we won't decrease the source-sink flow.

For randomized flow rounding, to round a flow we can again link the sink to the source with an edge $e$ that has flow value $F$, creating a circulation. After rounding the circulation, the expected flow along $e$ will be $F$, so the source-sink flow is preserved in expectation. Further, the algorithms in this paper change the flow along each edge by at most $1$ unit, so after rounding the source-sink flow will be between $\floor F$ and $\ceil F$.

%% file: Rounding.tex
\newpage
\section{Flow Rounding using Cycle Cancelling}
\label{sec:rounding}

\subsection{Fractional Cycles}

The following key lemmas motivate solving flow rounding using cycle cancelling. We assume that we are working with a graph with integral capacities throughout.

\begin{lemma}
	\label{lemma:forest}
	If a circulation $f$ has the property that the subgraph of all fractional edges forms a forest, then $f$ is an integral circulation.
\end{lemma}

\begin{proof}
	Suppose not. Then there must be a leaf node $v$ in the forest with only one connected edge that has fractional flow. The net flow into $v$ cannot be zero, as its remaining edges have integral flow, which violates the net flow condition on a circulation.
\end{proof}

\begin{lemma}
	\label{lemma:nocycle}
	If a circulation $f$ contains no cycle of edges with fractional flow, then $f$ is integral.
\end{lemma}

\begin{proof}
	A graph without cycles is a forest, therefore by Lemma \ref{lemma:forest} $f$ is an integral circulation.
\end{proof}

Now, suppose we begin with a fractional circulation $f$, from which we will construct a new circulation $f'$ (initially equal to $f$). If we cancel all fractional cycles in $f'$, then $f'$ must be integral. To cancel a given cycle, we can push flow around it until the flow across an edge becomes integral.

To capture the amount of flow required to make an edge integral, we define the \emph{availability} of a directed edge $(u,v)$ as:
$$ availability(u,v) = \ceil{f(u,v)} - f'(u,v) $$
We can push exactly $availability(u,v)$ units of flow along $(u,v)$ before $f'(u,v)$ becomes integral. This definition applies for both directions along the edge. For example, if $f'(u,v) = f(u,v) = 1.7$, then $availability(u,v) = 0.3$ and $availability(v,u) = (-1) - (-1.7) = 0.7$, indicating we can push $0.3$ units of flow from $u$ to $v$ or $0.7$ units of flow from $v$ to $u$ before the flow along the edge becomes integral. Similarly, we define the availability of a directed path (or directed cycle) to be the \emph{minimum availability} of edges along the path. This is the most flow we can push along a path before an edge becomes integral.

With the concept of availability, the next lemma guarantees that we can cancel cycles without violating capacity constraints.

\begin{lemma}
	\label{lemma:availability}
	If an edge $(u,v)$ satisfies $availibility(u,v) \geq 0$ and $availability(v,u) \geq 0$, then the new circulation $f'$ obeys the capacity constraints of the edge.
\end{lemma}

\begin{proof}
	If so, we have $\ceil{f(u,v)} - f'(u,v) \geq 0$ and $\ceil{f(v,u)} - f'(v,u) \geq 0$. Using antisymmetry, this gives us $\ceil{f(u,v)} \geq f'(u,v)$ and $\floor{f(u,v)} \leq f(u,v)$.
	As the capacity of $(u,v)$ is an integer, any flow value between $\ceil{f(u,v)}$ and $\floor{f(u,v)}$ must satisfy its capacity constraint, given that $f(u,v)$ does.  So $f'$ obeys the capacity constraint.
\end{proof}

If we only ever push flow along paths that is equal to the availability of the path, no edge can have its availability drop below zero. So the cycle cancelling method just described will yield an integral circulation $f'$ that obeys the capacity constraints.

\subsection{Costed Flow Rounding}

In the problem of \emph{costed flow rounding}, edges have associated costs per unit flow. Given a fractional circulation $f$, we want to find an integral circulation $f'$ that has no worse cost than $f$. This is achieved by exploiting the choice of direction we can make when cancelling a cycle: we can make an edge on the cycle integral by pushing flow in either direction, so we can choose the direction that yields better cost. As cost is antisymmetric, either one direction has positive cost and the other has negative cost, or both directions have zero cost. In either case, we can always cancel the cycle without increasing the cost of the circulation.

\subsection{Randomized Flow Rounding}

In our characterization of the \emph{randomized flow rounding} problem, we are given a fractional circulation $f$ with source $s$ and sink $t$, and wish to randomly find an integral circulation $\mathbf f'$ such that for every edge $(u,v)$, $\Ex[\mathbf f'(u,v)] = f(u,v)$ \footnote{$\Ex[\mathbf X]$ will denote the expected value of a random variable $\mathbf X$, and $\Pr[E]$ will denote the probability of an event $E$. Random variables are bolded.}. We say that a procedure \emph{preserves flow in expectation} if this is true.

One known algorithm for randomized flow rounding is the \emph{path stripping} algorithm originally proposed by \cite{Raghavan1985}. Their analysis of using path stripping to approximate integer multicommodity flow problems holds true for any algorithm that solves the randomized flow rounding problem, as it relies solely on the fact that the flow is preserved in expectation. The path stripping algorithm runs in $O(m^2)$ time, and we can improve that with another adaptation of cycle cancelling. The idea will be to randomly choose the direction in which to cancel each fractional cycle, such that the flow is preserved in expectation. First, we show that we can compose operations whilst preserving the flow in expectation.

\begin{lemma}
	\label{lemma:expectation}
	Let $\mathbf f$ be a random variable giving an initial flow. If $\rho$ is a procedure that preserves flow in expectation, then $\Ex[\rho(\mathbf f)] = \Ex[\mathbf f]$.
\end{lemma}

\begin{proof}
	Because $\rho$ preserves flow in expectation, $\Ex[\rho(g) - g] = 0$ for a fixed flow $g$. Then,
	\begin{align*}
		\Ex[\rho(\mathbf f) - \mathbf f]
			& = \sum_g \Ex[\rho(\mathbf f) - \mathbf f | \mathbf f = g] \Pr[\mathbf f = g] \\
			& = \sum_g \Ex[\rho(g) - g] \Pr[\mathbf f = g] \\
			& = 0
	\end{align*}
\end{proof}

\begin{lemma}
	\label{lemma:preservation}
	If $\rho_1, \rho_2, \hdots, \rho_n$ are procedures that preserve flow in expectation, then so does their composition.
\end{lemma}

\begin{proof}
	Let $f_0$ be the original flow, $\mathbf f_1 = \rho_1(f_0), \mathbf f_2 = \rho_2(\mathbf f_1), \hdots, \mathbf f_n = \rho_n(\mathbf f_{n-1})$. \\
	By Lemma \ref{lemma:expectation}, $\Ex[\mathbf f_{i+1}] = \Ex[\mathbf f_i]$, so $\Ex[\mathbf f_n] = f_0$.
\end{proof}

Now, if we can show how to cancel a cycle whilst preserving flow in expectation, any sequence of these will preserve flow in expectation. Suppose that a given cycle has availability $a$ forward and $b$ backward. That is, if we cancel it by pushing flow forward we will add $a$ units of flow to every edge, and pushing flow backward will subtract $b$ units. Fix a particular edge on the cycle, initially with flow $x$, and with final flow $\mathbf X$. If we cancel forward with probability $p$, then $\Ex[\mathbf X] = x + pa - (1-p)b = x + p(a + b) - b$. So, choosing $p = \frac {b} {a + b}$, $\Ex[\mathbf X] = x$. Note that $p$ is independent of $x$: we can choose $p$ independently of the flow on the particular edges of the cycle, and every edge will have its value preserved in expectation. So by Lemma \ref{lemma:preservation}, we can repeatedly cancel cycles in this way to obtain an integral flow whilst preserving the flow in expectation, solving the randomized rounding problem.

As such, our $O(m \log \frac{n^2}{m})$ cycle cancelling algorithm directly improves on the $O(m^2)$ path stripping approach for randomized rounding, with applications in global routing \cite{Raghavan1985}.

%% file: Rounding_mlgn.tex
\newpage
\section{Rounding in \texorpdfstring{$O(m \log n)$}{O(m log n)}}
\label{sec:mlogn}

In this section, we present our appropriation of Sleator and Tarjan's algorithm for making a flow acyclic \cite{Sleator1981} into an $O(m \log n)$ algorithm for cycle cancelling. We present the algorithm in the context of rounding a costed circulation, but any method for cancelling cycles that only needs path aggregation can be handled. For flow rounding, the algorithm adds fractional edges one by one and cancels cycles as they occur.

Initially we are given a fractional circulation $f$ over a costed graph $G = (V, E)$, and initialize our new circulation as $f' = f$.
We will build a graph $G' = (V, E')$ incrementally, initially with no edges. We initialize a dynamic trees data structure, with every node as a single-node tree, that will represent the fractional edges in $G'$ at all times. The data structure will keep track of the cost of each edge, flow along each edge, and availabilities in both directions.

We proceed to add edges $(u,v) \in E$ with fractional flow $f'(u,v)$ to $G'$ one by one. If $FindRoot(u) = FindRoot(v)$, a path from $u$ to $v$ exists and we have a fractional cycle (as in figure \ref{fig:tree1}). We find the sum of costs along the cycle, pick a direction in which cost is non-negative, and find the \emph{minimum availability} in that direction. Then we push this amount of flow around the cycle by adding to the $u$-$v$ path and $f'(u,v)$.

Now there must be at least one edge on the cycle with zero availability, and hence integral flow (as in figure \ref{fig:tree2}). If $f'(u,v)$ is integral, we simply don't add it to the dynamic trees. If any edge on the $u$-$v$ path is integral, a minimum availibility query will yield an edge $(x, y)$ with zero availability. We update $f'(x, y)$ with the flow stored in the data structure, and remove $(x,y)$ (figure \ref{fig:tree3}). If any integral edges remain, we can find them with a minimum availability queries along the $u$-$x$ path and the $y$-$v$ path, and remove them recursively.

\begin{figure}[H]
	\centering
	\begin{subfigure}{0.3\textwidth}
		\includegraphics[width=\textwidth]{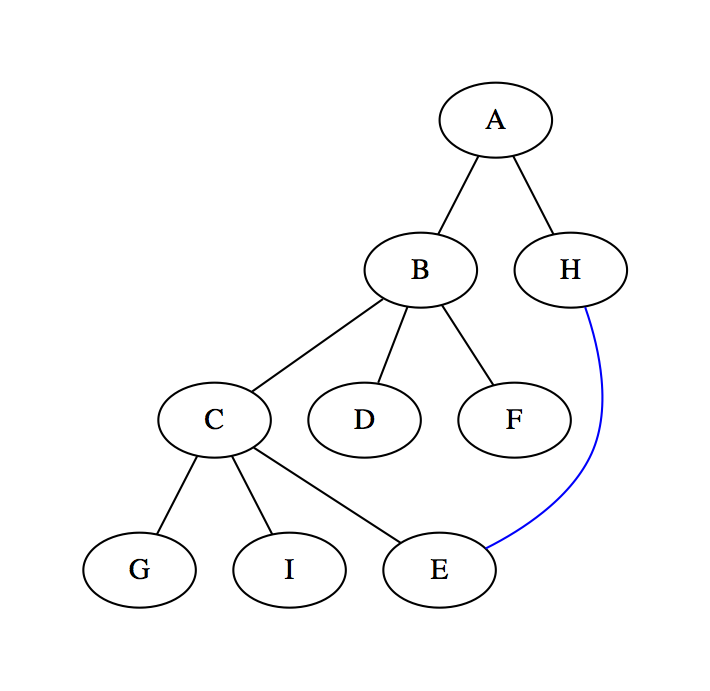}
		\caption{H and E are in the same tree, so adding (H,E) forms a cycle.}
		\label{fig:tree1}
	\end{subfigure}
	\hspace{10px}
	\begin{subfigure}{0.3\textwidth}
		\includegraphics[width=\textwidth]{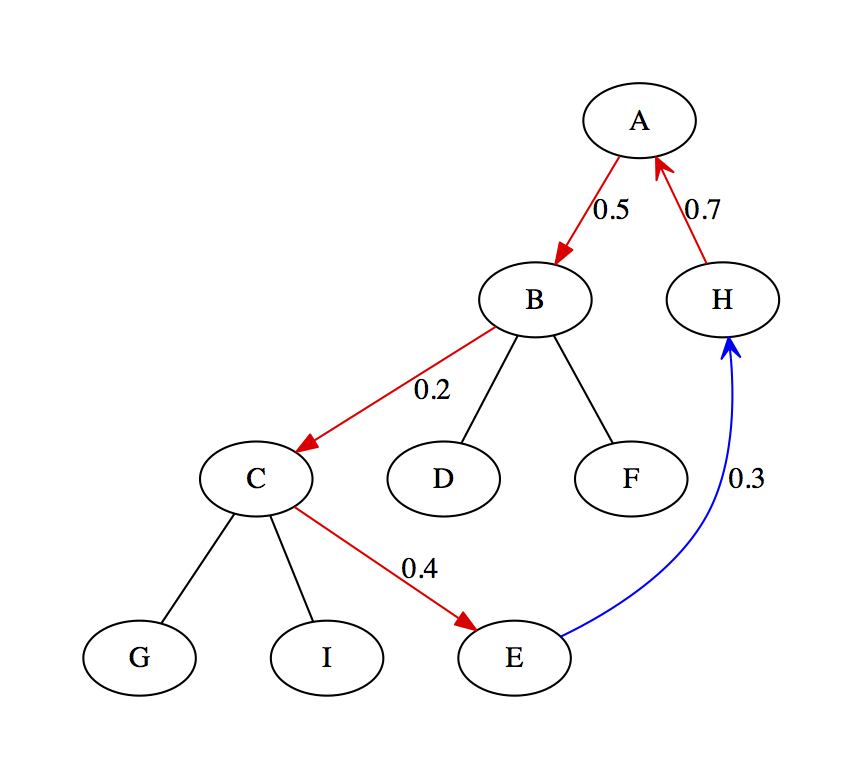}
		\caption{(B,C) has the smallest availability, so we push $0.2$ units around.}
		\label{fig:tree2}
	\end{subfigure}
	\hspace{10px}
	\begin{subfigure}{0.3\textwidth}
		\includegraphics[width=\textwidth]{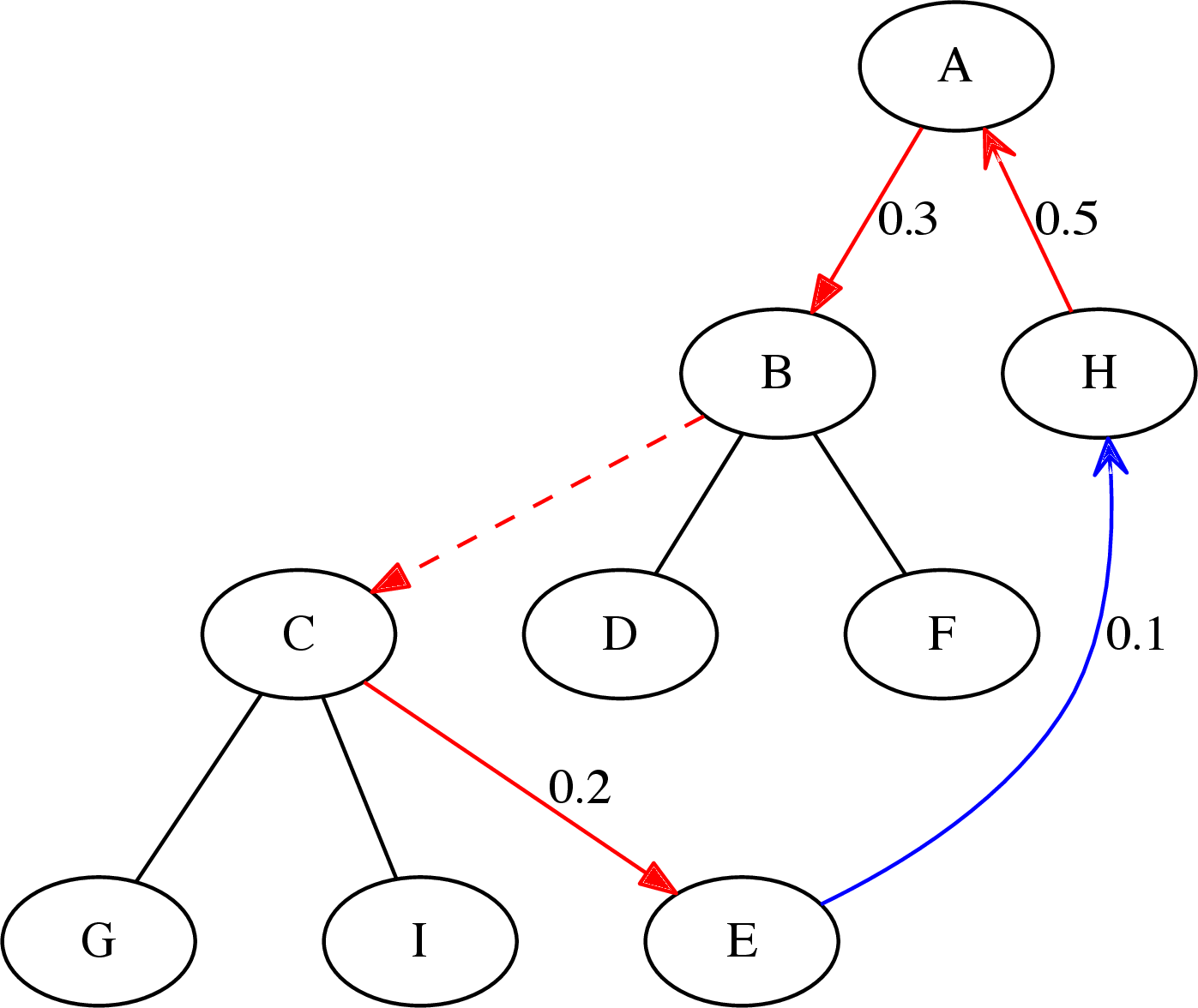}
		\caption{Decreasing availabilities, (B,C) becomes integral and is cut.}
		\label{fig:tree3}
	\end{subfigure}
	\caption{Cancelling a cycle in the $O(m \log n)$ algorithm.} \label{fig:mlogn}
\end{figure}

After adding all the edges to $G'$, there are no cycles of fractional edges, and so by \ref{lemma:nocycle} $f'$ must be an integral circulation. Further, the availability of each edge never moves below zero in either direction, so by \ref{lemma:availability} $f'$ satisfies all capacity constraints. Each dynamic trees operation used take $O(\log n)$ time, and we use a constant number of operations to add and remove each edge. As such, the total running time is $O(m \log n)$.

%% file: Rounding_n2.tex
\clearpage
\section{Rounding in \texorpdfstring{$O(n^2)$}{O(n2)} }
\label{sec:n2}

In the $O(m \log n)$ algorithm, we processed edges in an arbitrary order. In this section we show how to process all edges from the same node in one batch, cancelling all the cycles introduced at once in $O(n)$ time per node.

We maintain a forest of processed nodes initially empty. To process a node $x$, we consider its fractional edges that are connected to trees in the forset. If $x$ has more than $1$ edge to a tree, then cycles will be formed. We describe a recursive algorithm, $Cancel(u)$, that removes cycles involving the subtree of $u$. After invocation, we guarantee that there is at most one path to $x$ through the subtree left.

To perform $Cancel(u)$, we first call $Cancel(v)$ for each child $v$ of $u$. Each call to a child will cancels cycles, and returns information about the single remaining path to $x$ if exists. If more than one of $u$ and its children have paths left to $x$, there will be a cycle consisting of two disjoint paths to $x$. Let the paths be $pathDown$ and $pathUp$, such that $pathDown$ is a $u$-$x$ path, $pathUp$ is an $x$-$u$ path.

At this point, we can cancel the cycle as before based off the aggregate information. In the case of costed flow rounding, we find $cost(pathDown) + cost(pathUp)$, and swap the two if the cost is positive. Then we will send flow $F$ equal to the minimum availability of $pathDown$ and $pathUp$ around the cycle, down $pathDown$, and up $pathUp$, and update the availabilities of each path. (Actually updating the flow values on each edge is deferred until later.)

After cancelling a cycle, at least one of $pathDown$ and $pathUp$ will now have zero availability in one direction. We remove such paths (we will remove the integral edges later). We repeat this until at most one path from $u$ to $x$ remains, and return the new aggregate information. In particular, if $p$ is the parent of $u$, we can compute the new aggregates for the path $p$-$x$ using information about the edge $(p, u)$ and the path $p$-$x$.

To send flow around all of the cycles through $x$ in $O(n)$ time, we push flow along paths in the tree in a batch operation. If we're at a node $u$ and want to push $F$ units of flow down some $pathDown$ and up some $pathUp$, we immediately update the edges between the paths and $x$, mark the last node of $pathDown$ with $-F$, and mark the last node of $pathUp$ with $F$. After running $Cancel$, we run a procedure $UpdateFlow(u)$ on the root of each tree that does the following:
\begin{itemize}
	\item Let $incomingFlow$ initially be the sum of the marked values of $u$.
	\item {
		For each child $v$ of $u$:
		\begin{itemize}
			\item Call $UpdateFlow(v)$. Let the result be $F_v$.
			\item Subtract $F_v$ from $f'(u, v)$.
			\item If $f'(u, v)$ is now integral, remove it from the forest of fractional edges.
			\item Add $F_v$ to $incomingFlow$.
		\end{itemize}
	}
	\item Return $incomingFlow$.
\end{itemize}

At this point, we can add $x$ to the forest, as no cycles remain. Each above procedure takes time linear in the size of the forest, so adding every node takes $O(n^2)$ time. At termination, the fractional edges must form a forest, so by \ref{lemma:nocycle}
the circulation is integral.

Figure \ref{fig:n2} illustrates some of the steps in processing a node $X$ with edges into the existing tree.
The figures \ref{fig:n2} (a)(b)(c) depict $Cancel(H)$, and the figures \ref{fig:n2} (d)(e) illustrate $Cancel(C)$.

\newpage
\newgeometry{left=2.9cm,right=3.1cm,top=1cm,bottom=3cm}
\begin{figure}[ht]
	\centering
	\begin{subfigure}{0.49\textwidth}
		\includegraphics[width=\textwidth]{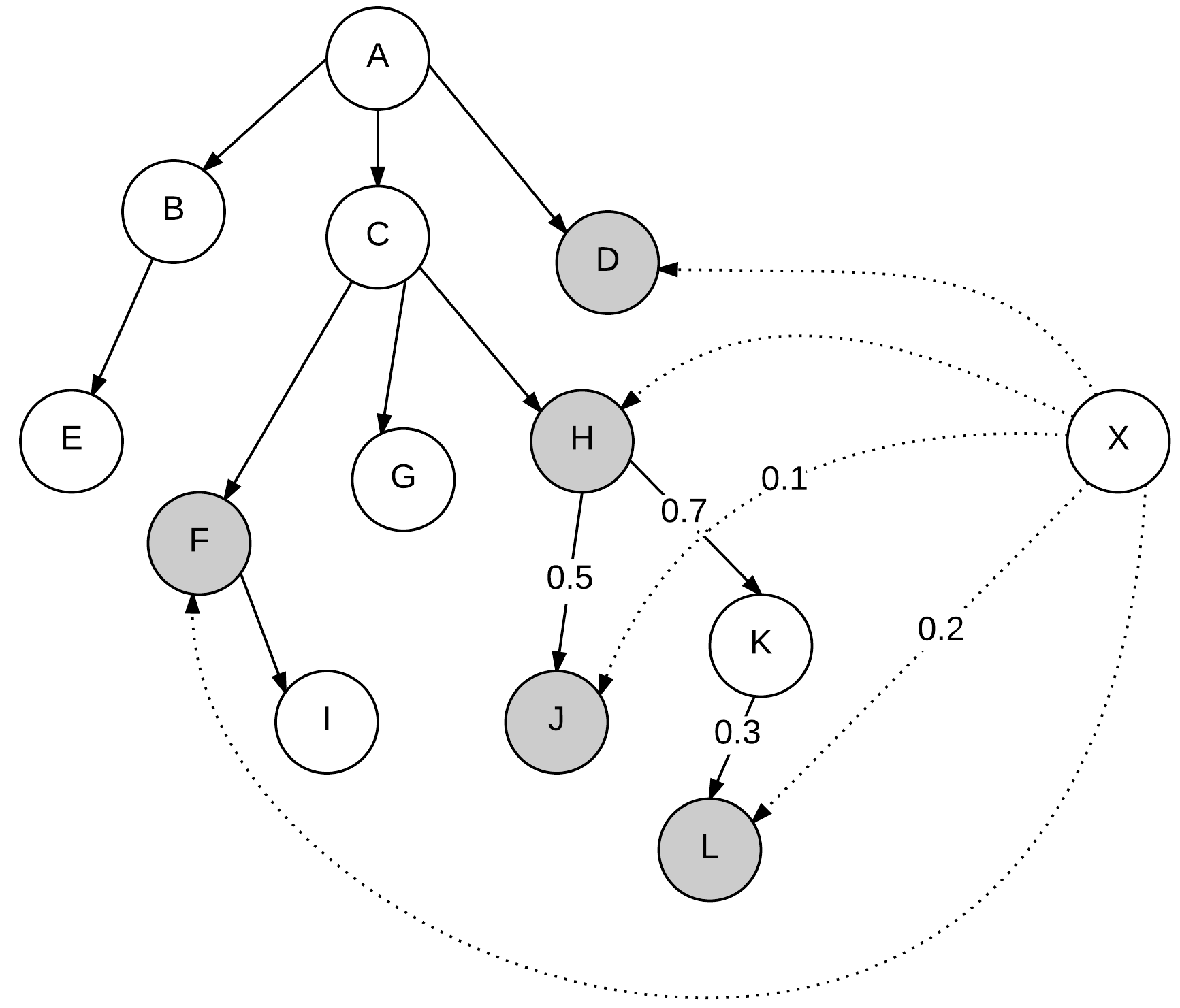}
		\caption{In $Cancel
	(H)$, cancelling the labelled cycle.}
		\label{fig:n2-1}
	\end{subfigure}
	\hspace{\fill}
	\begin{subfigure}{0.49\textwidth}
		\includegraphics[width=\textwidth]{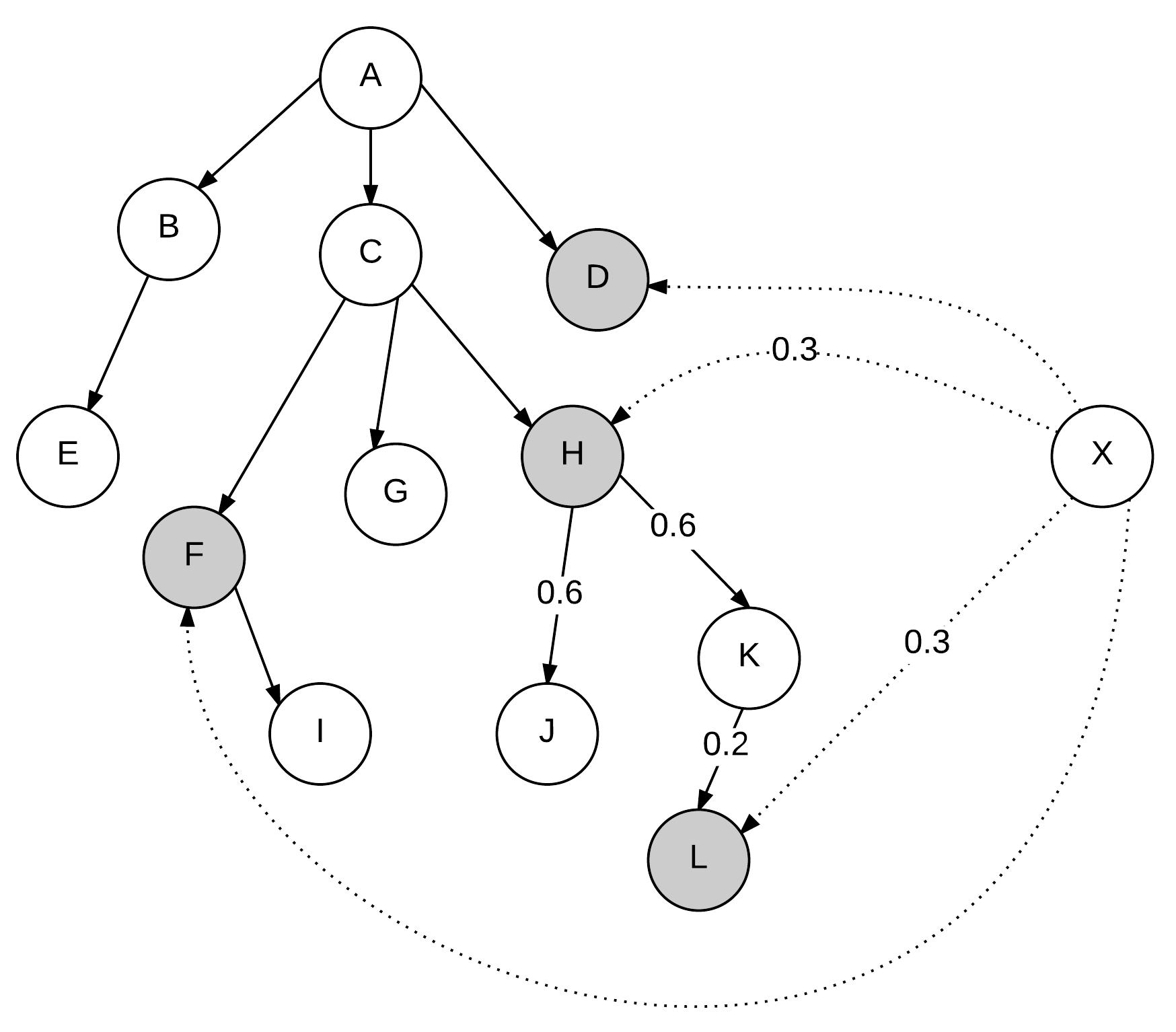}
		\caption{In $Cancel
	(H)$, cancelling the cycle through $L$.}
		\label{fig:n2-2}
	\end{subfigure}

	\vspace{5pt}

	\begin{subfigure}{0.49\textwidth}
		\includegraphics[width=\textwidth]{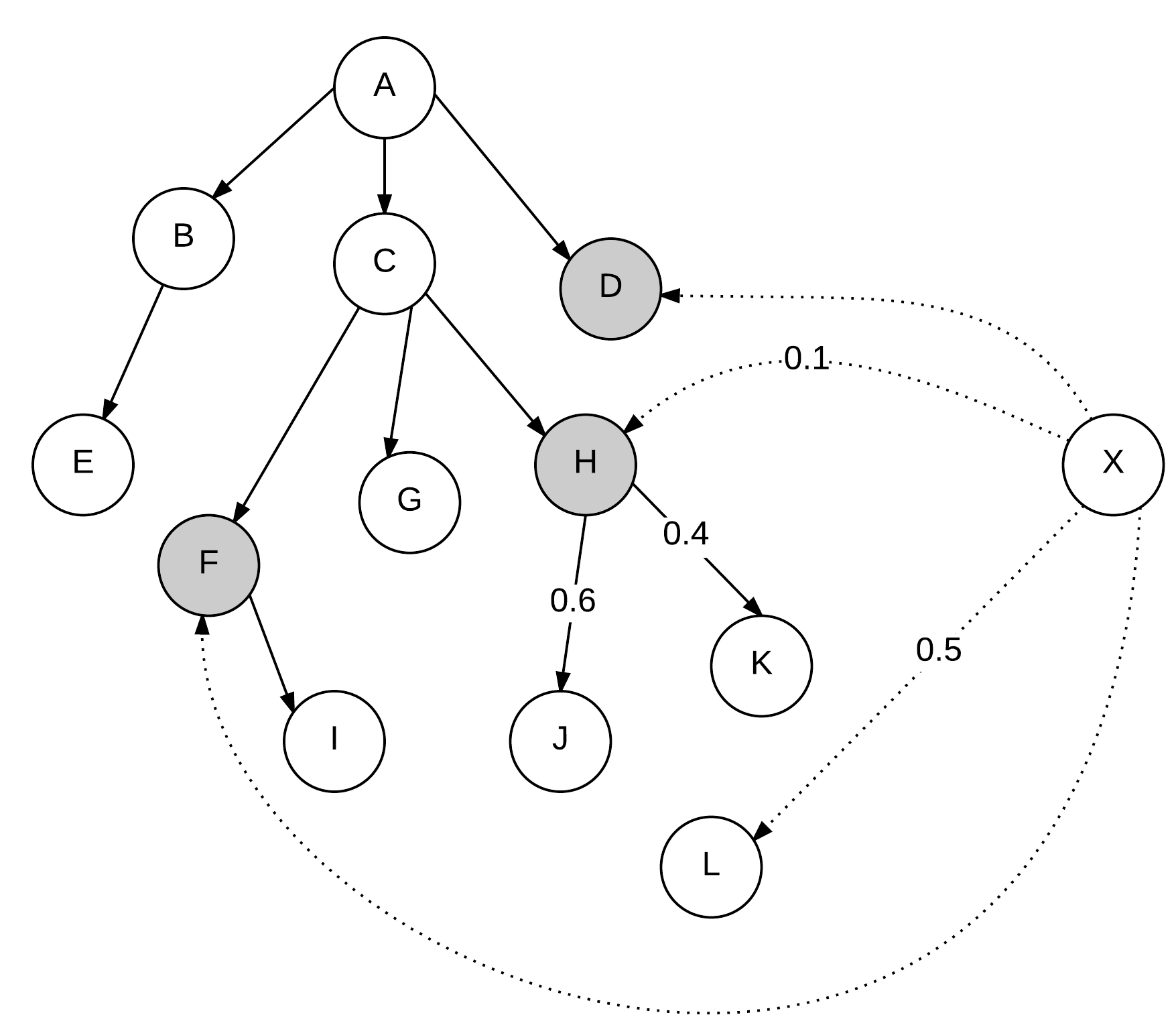}
		\caption{The result when $Cancel
	(H)$ returns.}
		\label{fig:n2-3}
	\end{subfigure}
	\hspace{\fill}
	\begin{subfigure}{0.49\textwidth}
		\includegraphics[width=\textwidth]{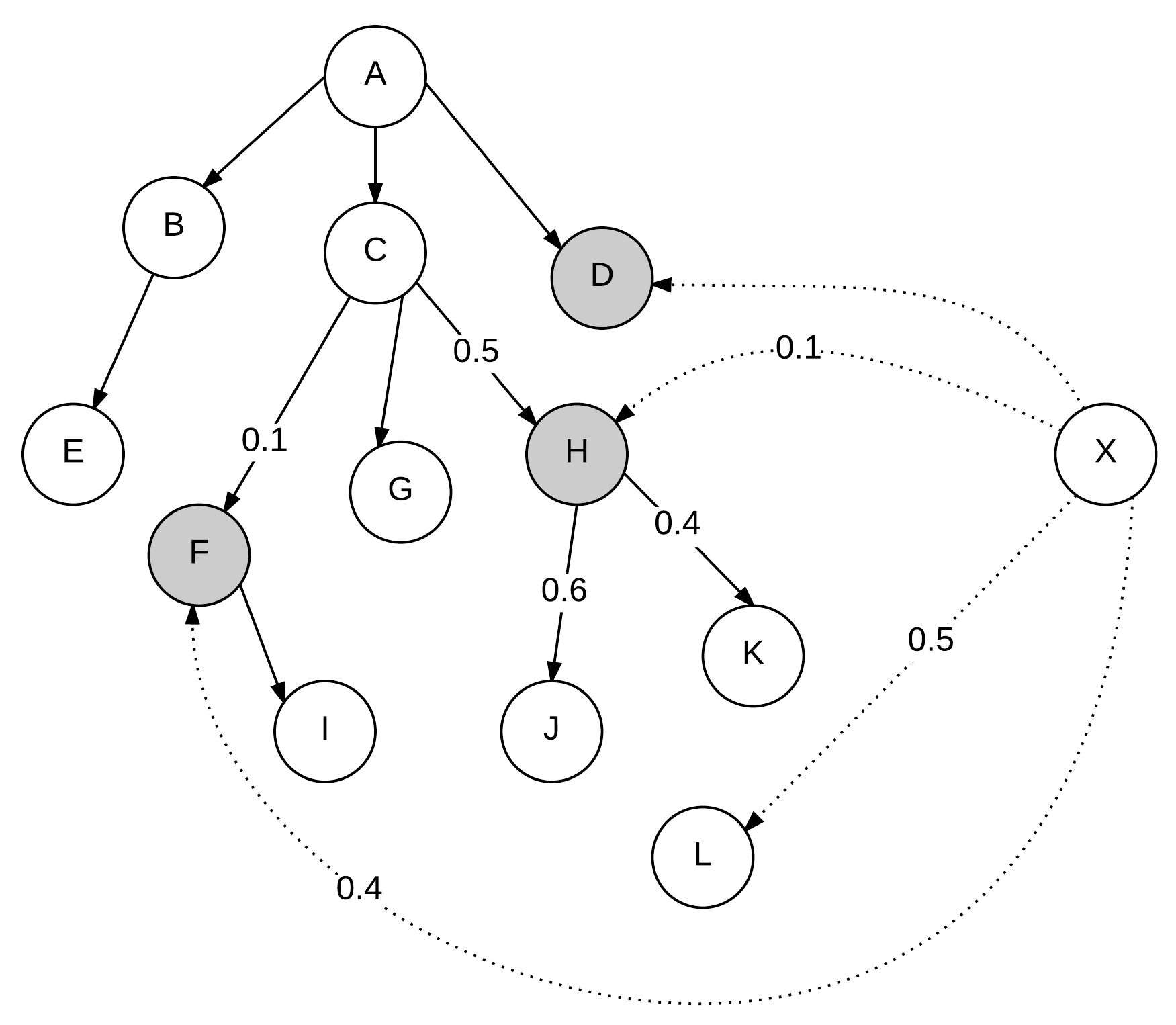}
		\caption{$Cancel
	(C)$ now cancels the cycle through $F$ and $H$.}
		\label{fig:n2-4}
	\end{subfigure}

	\vspace{5pt}

	\begin{subfigure}{0.49\textwidth}
		\includegraphics[width=\textwidth]{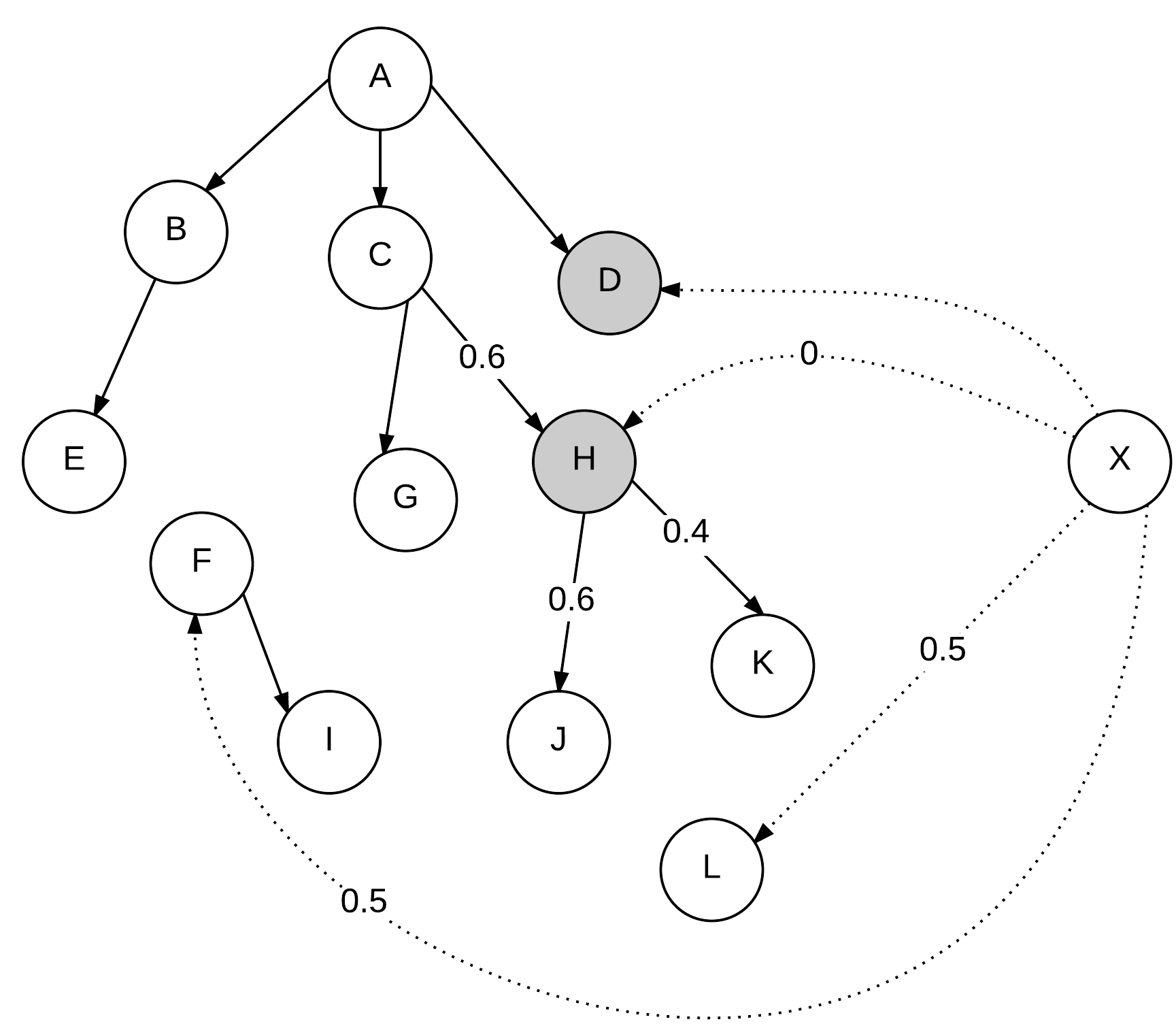}
		\caption{The result when $Cancel
	(C)$ returns.\\ Note $available(X,H) = 0$.}
		\label{fig:n2-5}
	\end{subfigure}

	\caption{Illustration of the $O(n^2)$ algorithm.}\label{fig:n2}
\end{figure}
\restoregeometry

%% file: Rounding_mlgn2m.tex
\clearpage
\section{Rounding in  \texorpdfstring{$O(m \log \frac{n^2}{m})$}{O(m log (n2/m))}}
\label{sec:mlogn2m}

Combining the dynamic trees of the $O(m \log n)$ algorithm and the batch processing of the $O(n^2)$ algorithm, it's possible to achieve an $O(m \log \frac{n^2}{m})$ algorithm for cycle canceling.

The basic idea is to represent the forest of processed nodes maintained by the $O(n^2)$ algorithm as trees of \emph{clusters}, where each cluster has a limited number of nodes, say less than $2k$, and is represented using the dynamic trees data structure. We call a tree of clusters a \emph{primary tree}, and will represent the primary trees by storing parent pointers in the root of each cluster. Each pointer gives the parent \emph{node}, from which we can find the parent cluster. By trading off the size of the clusters against that of the primary trees we achieve the stated speedup.

To keep the size of clusters in the required range, whenever we walk up a primary tree we will merge clusters with their parents if both contain less than $k$ nodes. After such an operation, the following key properties hold:

\begin{enumerate}
    \item Each cluster has at most 2k nodes.
    \item For any two adjacent clusters, one must have more than $k$ nodes.
    \item The number of internal (non-leaf) clusters in the primary trees is $O(n/k)$.
	\begin{proof}
		Consider decomposing the primary trees into sets of paths with one leaf cluster per path.
		This can be done for a particular tree by: numbering leaf clusters from left to right; choosing the first path to be from the first leaf to the root; choosing successive paths to be from the next leaf cluster to its lowest common ancestor with the previous leaf cluster.

		Indexing all the paths found this way, let $x_i$ be the number of internal clusters in the $i$-th path (so $x_i + 1$ will be the path length). Let $n_i$ be the total number of real nodes (not clusters) contained in the $i$-th path. For every two adjacent clusters on a path, there are at least $k$ nodes. Thus $n_i \geq \floor{(x_i+1)/2}k \geq x_ik/2$. So $\sum{x_i} \leq n/k$, and thus the number of internal clusters of the primary trees is $O(n/k)$.
	\end{proof}
\end{enumerate}

Now we are in a position to describe the algorithm. As in the $O(n^2)$ algorithm, we proceed to add nodes one at a time to a processed set. Suppose we are adding the node $x$, and the number of edges from $x$ to the processed nodes is $d$.

\begin{description}
\item[Step 1] \emph{Merge clusters.}

We identify all the clusters with edges to $x$, and follow any parent pointers upward to discover the primary trees in which we might have introduced cycles. As we follow parent pointers, we merge adjacent clusters if both have less than $k$ nodes. After merges, we will have reached $O(d + n/k)$ clusters in the search - at most $O(n/k)$ internal nodes and $d$ leaves. If there were $s$ merges, then this takes $O((s + d + n/k)\log k)$ time (as we need $O(\log k)$ time to merge and traverse primary tree edges).

\item[Step 2] \emph{Process cycles within primary trees.}

Now we will use the DFS approach in our $O(n^2)$ algorithm on the primary trees found in Step 1. Note that each cluster in a primary tree can now have multiple edges to $x$, but the approach extends naturally. The algorithm needs to traverse each primary tree involved, and cancel one cycle for each of the $d$ edges. Canceling each cycle will use a constant number of path queries within the clusters, as will moving between nodes in the primary tree. As each query takes $O(\log k)$ time, this takes $O((d + n/k)\log k)$ time overall. Figure \ref{fig:mlogn2m} illustrates this step.

\item[Step 3] \emph{Update parent pointers and link clusters to $x$.}

After Step 2, we've updated the flow on all involved edges, and removed the parent pointers and edges in clusters that have integral flow. No cycles remain, as $x$ has at most one edge to each tree. To construct the new forest, we initialize a new cluster for $x$ and add parent pointers to it from the nodes $x$ is connected to. Now $x$ will be the new root for several trees, and we will need to reverse some of the parent pointers. These changes will involve changing the cluster roots of $O(d)$ clusters (those that $x$ remains connected to) and reversing at most $O(d + n/k)$ parent pointers, taking $O((d + n/k) \log k)$ time.
\end{description}

Overall, we add a node $x$ with degree $d$ and $s$ merges in $O((s+d+n/k) \log k)$ time. The total number of merges (over all nodes) and sum of degrees are $O(m)$, so the cost of adding every node is $O((m + n^2/k)\log k)$. Choosing $k = \frac{n^2}{m}$, this is $O(m \log \frac{n^2}{m})$.

\begin{figure}[H]
	\centering
	%
	\begin{subfigure}{\textwidth}
		\includegraphics[width=0.95\textwidth]{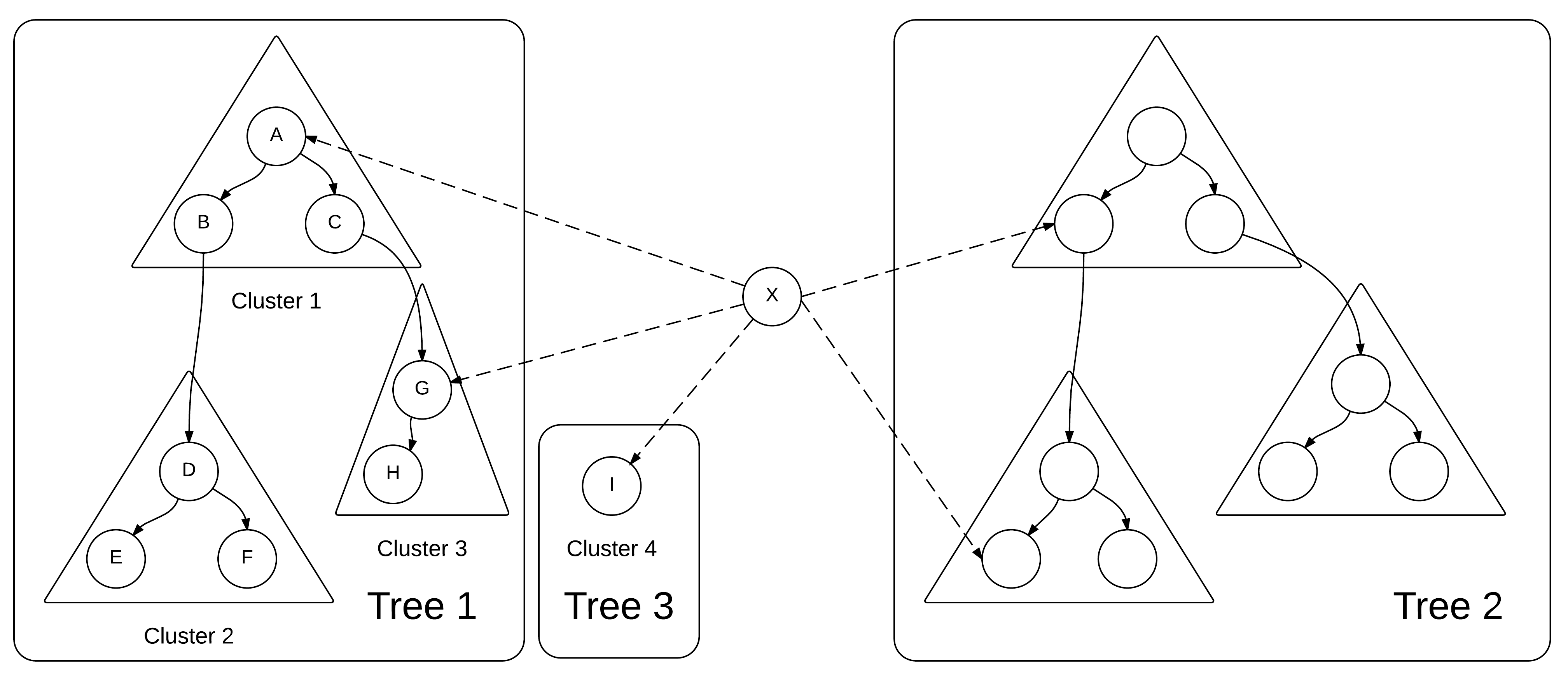}
		\caption{Before processing cycles within each tree.}
		\label{fig:step2}
	\end{subfigure}

	\begin{subfigure}{\textwidth}
		\includegraphics[width=0.95\textwidth]{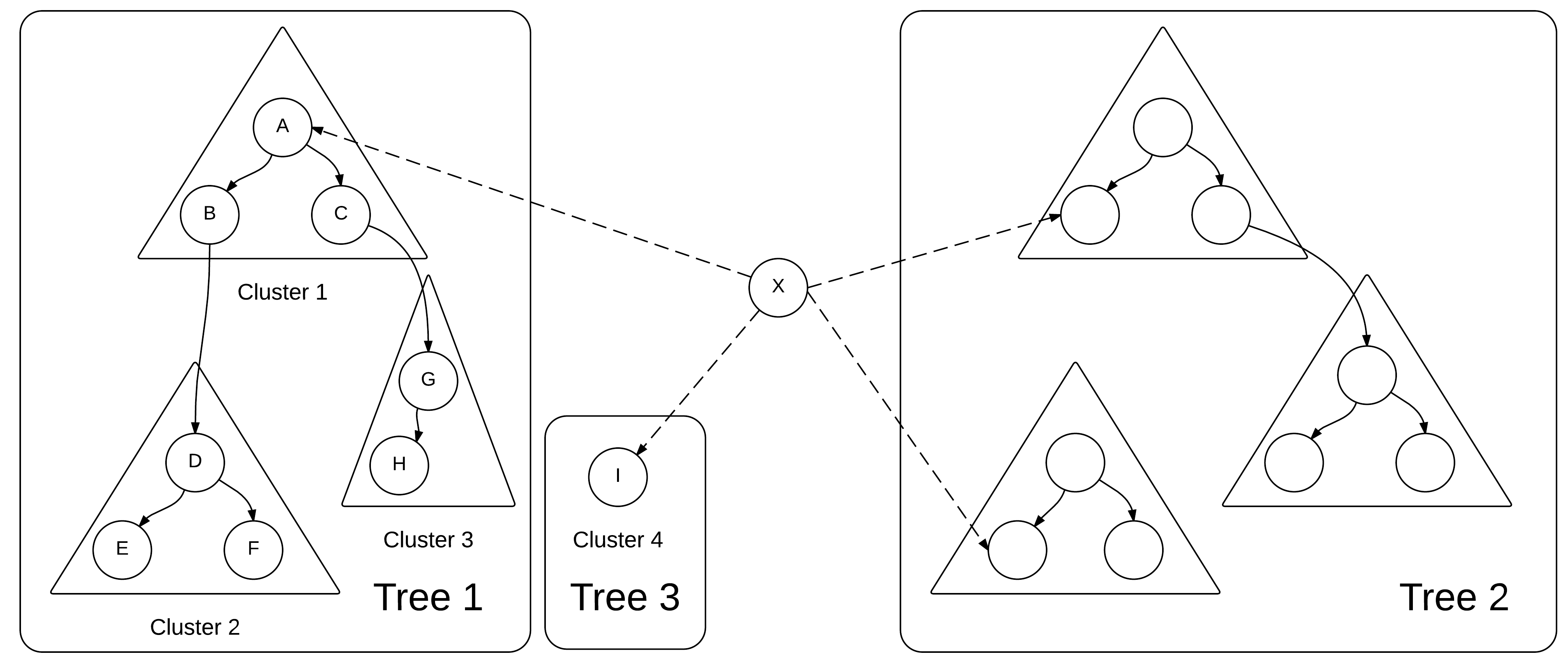}
		\caption{After cancelling cycles, Tree 2 is split, and $x$ has at most one edge to each tree.}
		\label{fig:step3}
	\end{subfigure}
	\caption{Step 2 of the $O(m \log (n^2/m))$ algorithm.}
	\label{fig:mlogn2m}
\end{figure}